\documentclass[sn-mathphys,Numbered]{sn-jnl}
\usepackage{graphicx,caption}%
\usepackage{multirow}%
\usepackage{amsmath,amssymb,amsfonts}%
\usepackage{amsthm}%
\usepackage{mathrsfs}%
\usepackage[title]{appendix}%
\usepackage{xcolor}%
\usepackage{textcomp}%
\usepackage{manyfoot}%
\usepackage{booktabs}%
\usepackage{algorithm}%
\usepackage{algorithmicx}%
\usepackage{algpseudocode}%
\usepackage{listings}%

\theoremstyle{thmstyleone}%
\newtheorem{theorem}{Theorem}
\newtheorem{lemma}[theorem]{Lemma}

\begin{document}

\title[Fractality in resistive circuits
]{Fractality in resistive circuits: The Fibonacci resistor networks}

\author*[1]{\fnm{Petrus} H. R. \sur{dos Anjos}}\email{petrus@ufcat.edu.br}
\equalcont{These authors contributed equally to this work.}
\author[2,3]{\fnm{Fernando} A. \sur{Oliveira}}
\equalcont{These authors contributed equally to this work.}
\email{faooliveira@gmail.com}
\author[3]{\fnm{David} L. \sur{Azevedo}}\email{david888azv@gmail.com}
\equalcont{These authors contributed equally to this work.}

\affil*[1]{Instituto de F\'isica, Universidade Federal de Catal\~ao, CEP 13560-970, Catal\~ao, GO - Brazil}

\affil[2]{Instituto de F\'isica, Universidade Federal Fluminense, CEP 24210-340, Niter\'oi, RJ - Brazil}
\affil[3]{Instituto de F\'isica, Universidade de Bras\'ilia, 70910-900, Bras\'ilia, DF - Brazil}

\abstract{
 We propose two new kinds of infinite resistor networks based on the Fibonacci sequence: a serial association of resistor sets connected in parallel (type 1) or a parallel association of resistor sets connected in series (type 2). We show that the sequence of the network's equivalent resistance converges uniformly in the parameter $\alpha=\frac{r_2}{r_1} \in [0,+\infty)$, where $r_1$ and $r_2$ are the first and second resistors in the network. We also show that these networks exhibit self-similarity  and  scale invariance, which  mimics a self-similar fractal. 
 We also provide some generalizations, including resistor networks based on high-order Fibonacci sequences and other recursive combinatorial sequences.}



\maketitle

\section{Introduction}

 The well known Fibonacci sequence $1,1,2,3,5,8\cdots$  intrigued and inspired
people through the centuries to delve more deeply into the recurring patterns of the physical world. In physics and chemistry, there is an increasing interest in the Fibonacci sequence since the discovery of quasicrystals \cite{bib:quasicrystal}.  

A Fibonacci number is obtained by a recursive procedure, namely adding the two previous Fibonacci numbers, i.e $f_1 = 1, f_2=1$ and $f_i = f_{i-1} + f_{i-2}$ thereafter. This idea was explored in other contexts. For example, a Fibonacci word is a specific sequence of symbols from any two-letter alphabet. The $n$th Fibonacci word is formed by repeated concatenation of two previous words, such that $W_n = W_{n-2} W_{n-1}$, so for the alphabet $A,B$, we have the sequence of words $A,B, AB, BAB, ABBAB, BABABBAB\cdots$  Another example, DNA segments can be arranged in a quasi-periodic Fibonacci sequence \cite{Azevedo16}. The evolution of a deterministic cellular automata with specific initial conditions can generate a Fibonacci fractal in the space-time. ~\cite{Devakul19,Anjos21} associated these fractals with the Fractal Symmetric Phases of Matter (FSPM). These FSPM appears in topological insulators\cite{Devakul19,Acosta19}. Recent results suggest that these strange objects are indeed very common in nature~\cite{Vergniory22}.

Resistance is in essence a response function which is a form of fluctuation-dissipation relation or fluctuation-dissipation theorem (FDT).  The FDT is a weak theorem and it fails in many situations such as in structural glass\cite{Grigera99,Crisanti03,Barrat98,Bellon02}, in proteins~\cite{Hayashi07}, in  mesoscopic radiative heat transfer~\cite{Perez-Madrid09,Averin10} and as well in ballistic diffusion~\cite{Lapas08,Costa03,Costa06,Lapas07}. For growth phenomena such as those described by the Kadar-Parisi-Zhang equation~\cite{Kardar86} different formulations of the FDT in $1+1$ dimensions have been proposed~\cite{Wio17,Rodriguez19,GomesFilho21}.
 For $2+1$ dimensions, the FDT fails~\cite{Kardar86}. Recentely, Anjos et al~\cite{Anjos21}  proposed that the growth dynamics builds up an interface with a fractal dimension $d_f$. This is particularly easy to see in the etching mechanism~\cite{Mello01,Rodrigues14,Alves16,Rodrigues24}, where an acid corrodes an initially crystalline surface creating a rough interface. This modifies the FDT, which should be now analyzed  in a fractal geometry and not in an Euclidean one~\cite{Anjos21}. This allowed a possible solution for the KPZ exponents. They have obtained for the roughness exponent~\cite{GomesFilho21b} $\alpha=(d_f+1)^{-1}$, what along with~\cite{Barabasi95} $\alpha=d-d_f$ yields, for $d=2$, $d_f=\varphi= \frac{\sqrt{5}+1}{2}$, i.e. the golden ratio. This allowed to obtain all the exponents~\cite{GomesFilho21b} and the fractal dimension-\cite{Luis22,Luis23}. Anomalous Diffusion is a Basic Mechanism for the Evolution of Inhomogeneous Systems\cite{Oliveira19}, such as diffusion in random fractals~\cite{Reis96,Reis24}. {Furthermore, the presence of fractals in phase transition has been well recognized~\cite{Coniglio89,Lima24}.}

Motivated by the aforementioned results, we study resistors networks constructed based upon a Fibonacci like sequence where each branch in series or parallel is a Fibonacci generation for each type of network. We can then obtain two kinds of networks: type $1$ is a serial association of resistors in parallel, see figure~\ref{fig:type1}, and type $2$ is a parallel association of resistors in series, see figure~\ref{fig:type2}. These type of network can be approximated in laboratory using for example polymers arranged  in a proper sequence. The interplay between the Fibonacci sequence, the golden ratio $\varphi$, and infinite ideal resistive circuits always have attracted attention. Currently, it spans all levels of academia from textbook exercises to research papers (e.g.  \cite{bib:Srinivasan} ). The Fibonacci Resistor Networks (FRN) presented here are an interesting and non-trivial problem that could be used in Physics and Mathematics teaching \cite{bib:Zemanian,bib:Atkinson}. Nevertheless, materials science and engineering applications are expected, specially since the FRN has a fractal-like pattern. Infinite circuits that have fractal-like patterns have been subject to extensive analysis due to their interesting  properties   \cite{bib:Chen,bib:Boyle}. Also, electrical properties of percolation clusters \cite{bib:Clerc} and the electric response of inhomogeneous materials can be investigated with fractal-like networks \cite{bib:Clerc2}.

\begin{figure*}
\begin{minipage}{0.3\linewidth}
\centering
    \includegraphics[scale=0.5]{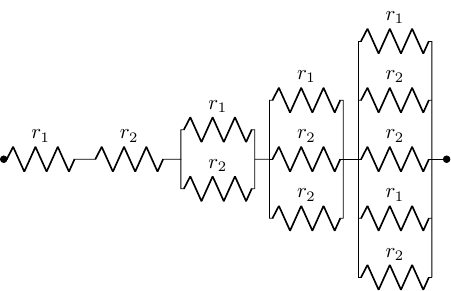}
    \captionsetup{width=\linewidth}
    \caption{Fibonacci resistor networks, type 1. Four Fibonacci generations of a serial association of resistors in parallel.}
    \label{fig:type1}
\end{minipage}
\hspace{0.02\linewidth}
\begin{minipage}{0.3\linewidth}
    \centering
    \includegraphics[scale=0.5]{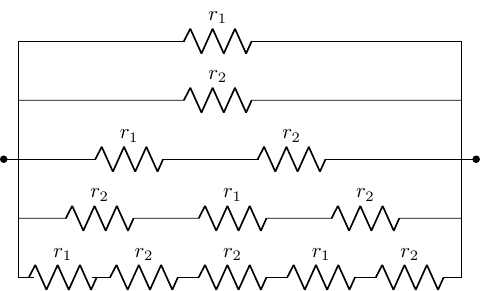}
    \caption{Fibonacci network, type 2. Four Fibonacci generations of parallel association of resistors in series.}
    \label{fig:type2}
\end{minipage}
\hspace{0.02\linewidth}
\begin{minipage}{0.3\linewidth}
    \centering
   \includegraphics[scale=0.5]{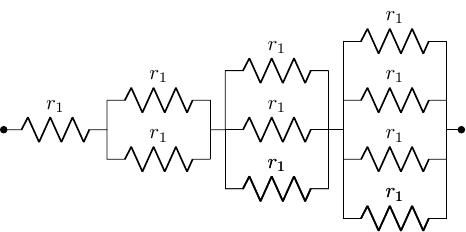}
    \caption{The EqR of this series network of parallel associations of an arithmetic progression of $r_1$ resistors  diverges logarithmically.}
    \label{fig:divergence}
\end{minipage}
\end{figure*}

All concepts investigated here could be easily extended to the elastic deformation of materials, where the elastic properties could be mimicked for a set of springs in serial or parallel association. Moreover, in some systems, such as in the fuse model, the rapture of spring bonds in an Euclidean lattice, creates a fractal structure at interface \cite{Arcangelis89,Hansen91}. As well in phenomena such as etching of interfaces \cite{Gomes19} and in growth models, one can start with a crystalline lattice and end up with a fractal structure~\cite{Anjos21,GomesFilho21b,GomesFilho24}.

Various aspects of both finite and infinite resistor networks have been studied previously. For example, \cite{ref1,ref2,ref3} studies the equivalent resistance (EqR) of all possible combinations of $n$ equal resistor, found lower and upper bounds for the EqR and a power law for the equivalent resistor independent of if this association are planar or not. In another examples, inifite resistor networks were studied using Lattice Green function methods \cite{Asad05,Asad06,hijjawi2008infinite,asad2013infinite}.  We show that the EqR of the Fibonacci circuits studied here is given by the series of generalized Fibonacci numbers' reciprocals. The interest in series of reciprocals dates back at least before Euler  \cite{bib:Euler}. It is worth noting that not all recursively generated resistor network lead to a convergent EqR; a simple example is an arithmetic resistor network shown in figure \ref{fig:divergence} with EqR diverges as $n\rightarrow + \infty$ since the harmonic series $\sum \frac{1}{n}$ diverges logarithmically.

In this paper, we analyze Type 1 and Type 2 FRN, calculate the EqR of each generation and show that the EqR of the entire network is proportional to the sum of the reciprocal of generalized Fibonacci numbers. We discuss how this FRN networks exhibit self-similarity (a fractal-like behavior) and how the EqR changes as the scale of the fractal changes. We discuss some possible generalizations to FRN, in particular, we discuss a p-order FRN and other resistor networks based on other recursive combinatorial sequences such as Pell sequence and Bell numbers.

\section{Type 1 Fibonacci Resistor Networks}\label{sec:type1}

In this section, we consider the properties of the Type 1 FRN. Let $r_2 = \alpha r_1$, and let $r_n$ denote the $n$-th generation EqR.  From figure \ref{fig:type1-ngen}, we see that for $\alpha \neq 0$
\begin{equation}\label{eq:type1-nth}
\frac{1}{r_n} = \frac{1}{r_{n-1}} + \frac{1}{r_{n-2}}.
\end{equation}
Writing $\frac{f_n}{r_1} = \frac{1}{r_n}$, we see that $f_n = f_{n-1} + f_{n-2}$ ($n \geqslant 3$), so the sequence of $f_n$ (the reciprocal of the $n$th EqR) is a Fibonacci sequence with $f_1 = 1$ and $f_2 = \alpha^{-1}$.

\begin{figure}[ht]
\begin{minipage}{0.45\linewidth}
    \centering
    \includegraphics[scale=0.6]{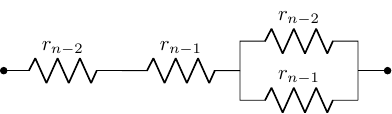}
    \caption{In a type 1 FRN the $n$th resistor is a parallel association of the two previous resistors.}
    \label{fig:type1-ngen}
\end{minipage}
\hspace{0.05\linewidth}
\begin{minipage}{0.45\linewidth}
\centering
  \includegraphics[scale=0.6]{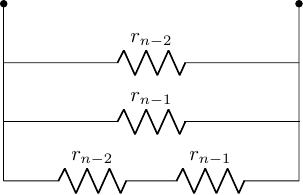}
  \caption{In a type 2 FRN the $n$th resistor is a serie association of the two previous resistors.}
 \label{fig:type2-ngen}
 \end{minipage}
\end{figure}

A Fibonacci sequence has a closed form for the $n$th term (as any sequence given by a linear recurrence relation). In fact, the ansatz  $f_n = Ax^{n-1}$, implies that $x$ is a root of $x^2 - x -1 = 0$. We have two roots
$$ \varphi = \frac{1 + \sqrt{5}}{2}\,\,\,\text{and}\,\,\,\,\psi = - \varphi^{-1} = \frac{1 - \sqrt{5}}{2}$$
and the linearity of the recurrence relation implies that
 \(f_n = A \varphi^{n-1} + B \psi^{n-1}\). We fix $A$ and $B$ to obtain $f_1=1$ and $f_2= \alpha$. Therefore, we have a Binet formula
\begin{equation}\label{eq:general-term-fibonacci}
f_n(\alpha) = \left[\frac{1}{\alpha ( \varphi - \psi )}\right] \left[ \left( 1 - \alpha \psi  \right) \varphi^{n-1} - \left(1 - \alpha \varphi \right)\psi^{n-1}\right].
\end{equation}
A quick inspection on this expression clearly shows that $f_n(\alpha) \rightarrow 0$ for $\alpha \neq +\infty$, so that $r_n(\alpha) \rightarrow 0$. Therefore $R_n^{(1)} (r_1,\alpha r_1)$, the EqR of the entire type 1 network up to the $n$th generation is 
$$ R_n^{(1)} (r_1,\alpha r_1) = \sum_{k=1}^n r_k (\alpha) = \sum_{k=1}^n \frac{r_1}{f_k (\alpha)} = r_1 {\cal F}_n (\alpha).$$
We will show that ${\cal F}_n$ converges uniformly for $\alpha \geqslant 0$. So we can write
$$ R^{(1)}_{eq} (r_1,r_2) = r_1 {\cal F} \left(\frac{r_2}{r_1}\right).$$
The $r_2=0$ case is trivial.

\section{Type 2 Fibonacci Resistor Networks}\label{sec:type2}

For a type 2 FRN, the equivalent resistor of the $n$th generation is $r_n = r_{n-1} + r_{n-2}$ so each generation is a Fibonacci sequence (see figure \ref{fig:type2-ngen}). For uniformity, let \( r_n = r_1 f_n\), so we get $f_1 = 1$ and $f_2 = \beta$, and the $n$-th term is given by Eq. \ref{eq:general-term-fibonacci} using $\alpha = \frac{1}{\beta}$.
The EqR such the network up to the $n$th generation is
$$ R^{-1}_n\left(r_1,\frac{r_1}{\beta}\right) = \sum_{k=1}^{n} \frac{1}{r_k} = \frac{1}{r_1} \sum_{k=1}^{n} \frac{1}{f_k} = \frac{1}{r_1}{\cal F}_n \left( \frac{1}{\beta} \right)$$

The EqR formula for Type 2 FRN follows from the Type 1 FRN formula. Consider the case where $\alpha$ is a positive rational number (so $r_2$ can be replaced by an association $r_1$ resistors). Given a circuit constructed from $n$ equal resistors $r_1$ in series and/or parallel combination that has an EqR $\beta r_1$, the configuration obtained by changing all serial connections to parallel and parallel connections to serial connections respectively,
results in a circuit, whose EqR is $\frac{1}{\beta} r_1$ (Theorem 1 of \cite{bib:Khan}). Now, the uniform convergence of ${\cal F}(\alpha)$ for $\alpha > 0$ can be used to extend this result to all $\alpha > 0$ since $\mathbb{Q}^+$ is dense in $\mathbb{R}^+$.

\section{Properties of ${\cal F}$} \label{sec:properties}

In this section, we obtain the properties of ${\cal F}$. It is convenient to write
\begin{equation}\label{eq:f-decomposition}
f_n(\alpha) = \frac{A_n}{\alpha} + B_n\end{equation} 
where
$$A_n = \frac{ \varphi^{n-1} - \psi^{n-1}}{( \varphi - \psi )}\,\,\,\, \text{and}\,\,\,\, B_n = \frac{\varphi \psi^{n-1} - \psi \varphi^{n-1} }{( \varphi - \psi )}.$$ 
Now, we define a generalized Fibonacci sequence (GFS) to be the sequence
\begin{equation}\label{eq:gibonacci}
g^{(x,y)}_1 = x,\,\,\,\,\,g^{(x,y)}_2 = y,\,\,\,\,\,g^{(x,y)}_n = g^{(x,y)}_{n-1} +   g^{(x,y)}_{n-2},\,\,\,\forall n\geqslant 3.    
\end{equation}
Fixing $x=1$, $y=1$, we recover the standard Fibonacci sequence $f_n(1)$.  The sequence $f_n(\alpha)$ is a GFS with $x=1, y=\frac{1}{\alpha}$. To analyse the convergence of ${\cal F}$ we need the following lemma:

\begin{lemma}\label{l:inequalities} The following statements holds
	\begin{enumerate}
		\item[(a)] $A_n= g_n^{(0,1)}$ and for $n\geqslant 1$, $A_{n} = f_{n-1}(1)$.
		\item[(b)] $B_n = g_n^{(1,0)}$  and for $n\geqslant 2$, $B_{n} = f_{n-2}(1)$.
		\item[(c)] $1 \leqslant\frac{A_n}{B_n} \leqslant 2$, $\forall n>2$.
		\item[(d)] Let $\Phi = \left\{ -\frac{A_k}{B_k} : k> 2 \right\}$, then $\Phi \subset [-2,-1]$.
	\end{enumerate}	
\end{lemma}
\begin{proof}
	(a) and (b) Follows from \ref{eq:f-decomposition} and  \ref{eq:gibonacci}.\\
		(c) and (d):  From (a) and (b),
		we just need to calculate an upper bounds for  
		 $1 \leqslant \frac{f_{n-1}(1)}{f_{n-2}(1)}$. Using that $f_{n-1}(1) = f_{n-2}(1)+f_{n-3}(1)$ and since $\frac{f_{n-3}(1)}{f_{n-2}(1)} \leqslant 1$, we find that
		 $1 \leqslant  \frac{f_{n-1}(1)}{f_{n-2}(1)} \leqslant 2$. (d) follows from (c).
\end{proof}

We now can apply the ratio test to analyze the pointwise convergence of the series ${\cal F}_n$. 

\begin{lemma}
	\(\sum_{k=1}^{+\infty} \frac{1}{f_k(\alpha)}\)
	converges pointwisely for all $\alpha \in \mathbb{R} - \Phi$.
\end{lemma}
\begin{proof}
	For $\alpha =0$, $f^{-1}_n = 0$ for $n\geqslant 2$, so ${\cal F}_n(0) \rightarrow 1$.	For $\alpha \not \in \Phi$, lemma \ref{l:inequalities}.(c) shows that $f_k \neq 0$ and we note that since $|\psi|<1$ then $ \psi^n \Rightarrow 0$ as $n \rightarrow + \infty$. So for $\alpha > 0$, we find that \( \frac{f^{-1}_{n+1}}{f^{-1}_{n}} = \frac{f_{n}}{f_{n+1}}  \Rightarrow  \frac{1}{\varphi} < 1 \) as $n\rightarrow 0$ and by the ratio test the series converges pointwisely for $\alpha \neq 0$.
\end{proof}

We now show that the convergence is uniform. 

\begin{theorem}\label{th:uniform} $\sum_{k=1}^{+\infty} \frac{1}{f_k(\alpha)}$ converges uniformly in $[0,+\infty)$.
	\end{theorem}
\begin{proof}
Note that 
$${\cal F}(\alpha) = 1 + \alpha + \sum_{k=2}^{+\infty} \frac{\alpha}{A_k + \alpha B_k}.$$
Clearly if $G(\alpha)=\sum_{k=2}^{+\infty} \frac{\alpha}{A_k + \alpha B_k}$ converges uniformly then ${\cal F}$ converges uniformly, since for $n,m \geqslant 2$, we have that  \(\left|{\cal F}_n - {\cal F}_m \right| = \left| G_n(\alpha) - G_m(\alpha)\right|\)
where $G_n(\alpha) = \sum_{k=2}^{+n} \frac{\alpha}{A_k + \alpha B_k}$. So we just need to show that $G$ converges uniformly. Now from Lemma \ref{l:inequalities}, $B_k > 0$, for all $k\geqslant 2$, so for $\alpha \geqslant 0$
$$ \frac{\alpha}{A_k + \alpha B_k} \leqslant \frac{1}{B_k}.$$
By the ratio test, the series $\sum_{k=2}^{+\infty} B_k^{-1}$ converges absolutely. So the sequence $G_n$ converges uniformly by Weierstrass M test \cite{bib:Rudin}, and therefore so does ${\cal F}$.
\end{proof}

Figure \ref{fig:f-initial} shows the behavior of the EqR of a Type 1 FRN. Note that for each $\alpha$ the graphic behaves as $1+\alpha$ for the initial generations and then approaches the convergence value. Futhermore, we see that ${\cal F}(\alpha) -(1+\alpha) \rightarrow  {\cal F}\left(\frac{1}{2}\right)$ as $\alpha \rightarrow +\infty$ as physically expected. The reasoning is the following, as $r_2 = \alpha r_1 \rightarrow +\infty $, then $r_3 \rightarrow r_1$ and $r_4 \rightarrow \frac{r_1}{2}$ so eliminating $r_1,r_2$ we remain with a Type 1 FRN with $\alpha=\frac{1}{2}$. A similar direct physical reasoning gives us the behavior of the $\alpha$ dependency of ${\cal F}^{-1} (\alpha^{-1})$ shown in figure \ref{fig:type2-result}. 

Our final result in this section, show that the EqR of the type 1 FRN is strictly increasing. Since ${\cal F}(\alpha)$ converges uniformly for $\alpha >0$, we have that 
$$ \frac{\mathrm{d}}{\mathrm{d}\alpha } {\cal F}(\alpha) = 1 + \sum_{k=3}^{+\infty} \frac{\mathrm{d}}{\mathrm{d}\alpha}\left(\frac{1}{f_k} \right) = 1 + \sum_{k=3}^{+\infty}   \frac{A_k}{\left( A_k + \alpha B_k \right)^2}\geqslant 1$$
since by Lemma \ref{l:inequalities}, $A_k \geqslant 0$.

\begin{figure*}[ht]
\begin{minipage}{0.45\linewidth}
    \centering
    \includegraphics[scale=0.4]{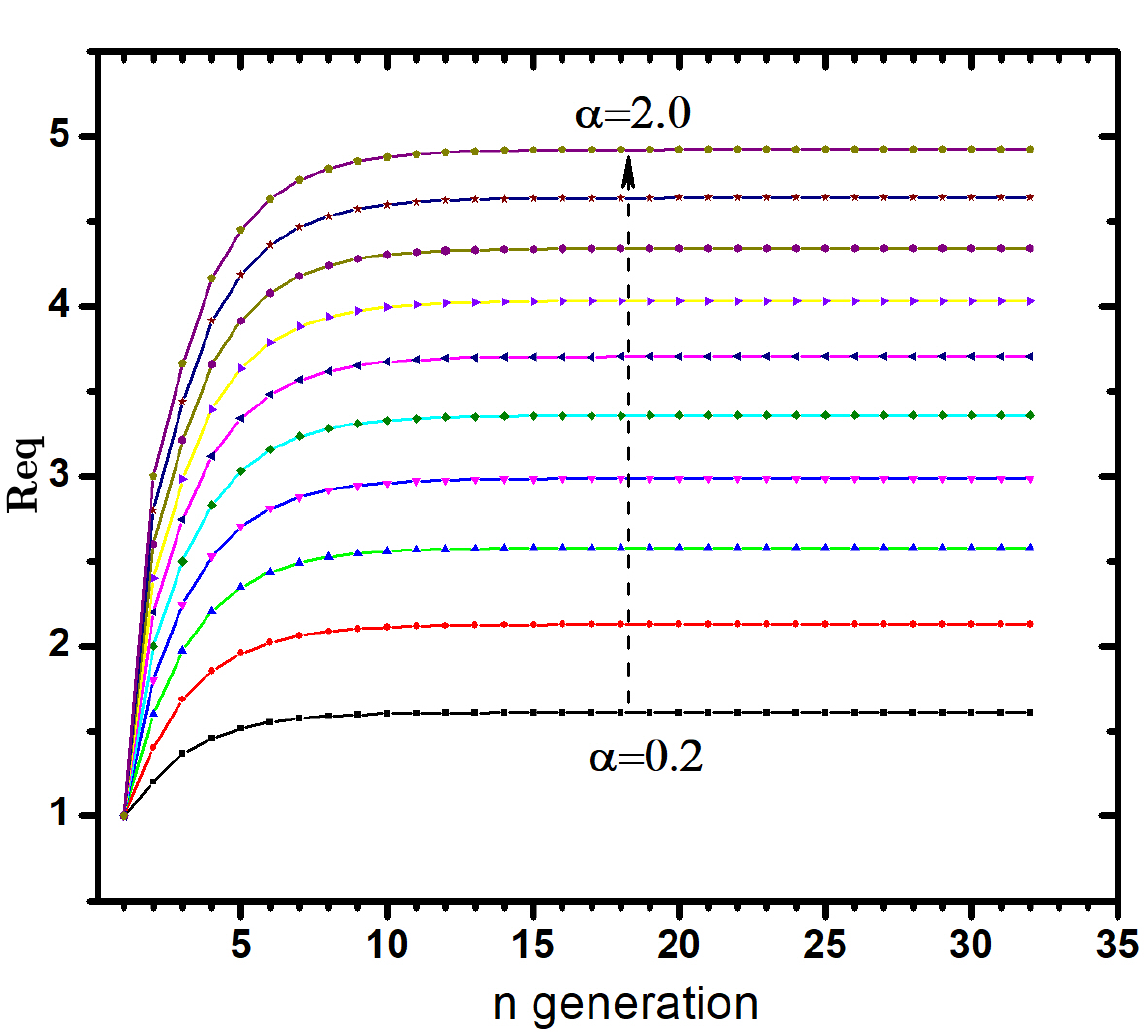}
    \caption{Behavior of the EqR of a Type 1 FRN with $r_1=1$.}
    \label{fig:f-initial}
\end{minipage}
\hspace{0.05\linewidth}
\begin{minipage}{0.45\linewidth}
    \centering
    \includegraphics[scale=0.4]{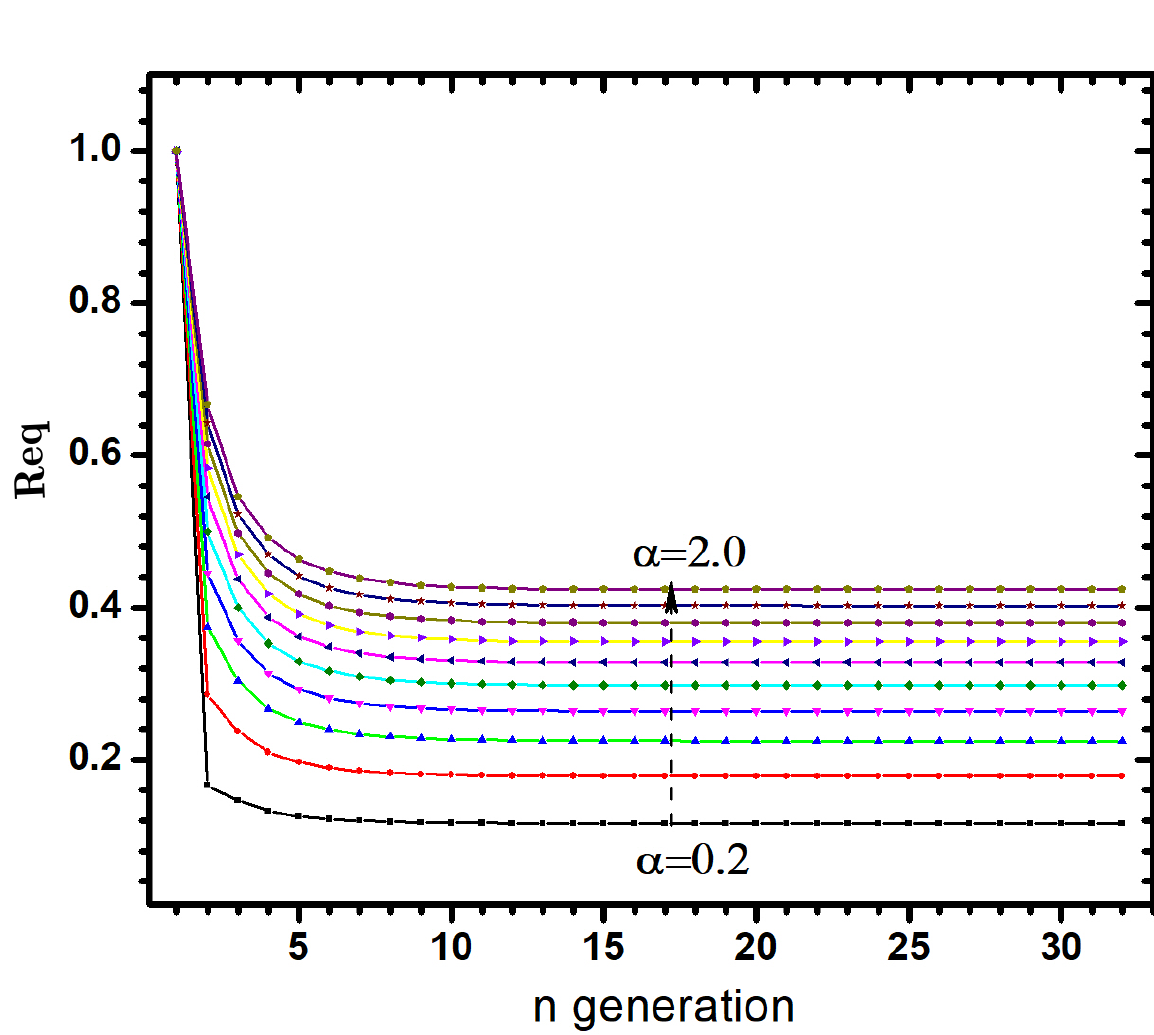}
    \caption{Behavior of the EqR of a type 2 FRN  as a function of $\alpha$ for $r_1 = 1$ is given by \({\cal F}^{-1} (\alpha^{-1})\)} 
    \label{fig:type2-result}
\end{minipage}
\end{figure*}

\section{Self Similarity}\label{sec:self-similar}

This section show that FRN exhibits self-similarity and scale invariance, which mimics a self-similar fractal. The idea of considering resistive circuits as fractal-like structures was illustrated in \cite{bib:Santos}. The critical point for this argument is that for an ideal resistive circuit, the only relevant data are the individual resistances of each resistor and the rule of interconnection among them. Another way to express this is to say that only topological information defines an ideal resistive circuit. Geometric aspects are not relevant in the construction of these models. That means ideal circuits are naturally represented as weighted graphs.

A simple undirected (multi-)graph $G$ is composed of a set $V$ of vertices and a set   $E(V)$ of edges between these vertices. In ideal resistive circuit graphs, the edges represent resistors between two vertices. The ideal resistive circuit graphs of interest are connected in the sense that for each pair of vertices in a connected graph, one can find a sequence of resistors connecting one vertex to the other. Each edge has a weight equals to the resistance that the edge represents. So each circuit can be represented by a graph $G(V,E,r)$, where $r: E \rightarrow \mathbb{R}^+$ is the weight function. Now, we focus on ideal resistive circuit graphs between two external points (or terminals) $A$ and $B$, we denote the set of such graphs by $G_{AB}$ and point out that the EqR gives us a natural measure of these graphs, so for $G(V,E,r) \in G_{AB}$ we write $\left|G(V,E,r)\right|$ to denote the EqR of the circuit $G(V,E,r)$.

For a type 1 FRN connecting points $A$ and $B$, we find that $V=\{A=1,2,3\cdots,+\infty=B\}$, $E = \{ E_{ij}^{(n)}: j=i+1, n=1,\cdot, F_i\}$, where $F_i$ is the $i$-th Fibonacci number, and $r\left(E_{ij}^{(n)} \right)$ is the $n$-th letter of the $i$-th Fibonacci word constructed using $r_1$ and $r_2$. We denote such weighted graph by $G_f(r_1,r_2)$. 

Now, consider the transformation $T$ that relabels  all $r_1, r_2 \rightarrow r_2,r_3$. It maps $G_f(r_1,r_2) \rightarrow G_f(r_2,r_3)$ as illustrated in figure \ref{fig:relabeling}. 
\begin{figure}[ht]
    \centering
    \includegraphics[scale=0.5]{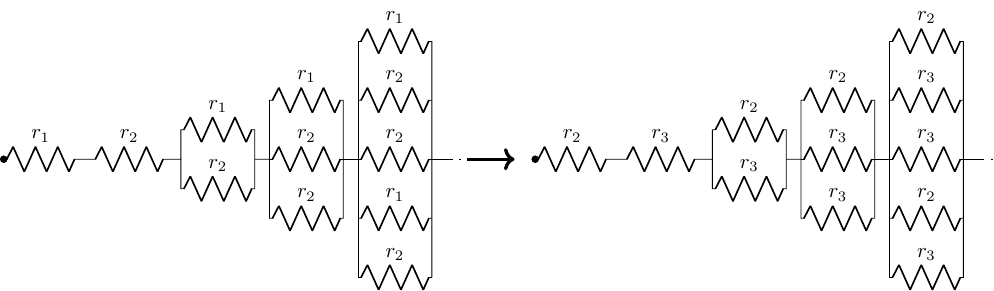}
    \caption{Relabeling type 1 FRN $r_1 \rightarrow r_2$, and $r_2 \rightarrow r_3$, we define a transformation $T$ that preserves the topology of the weighted graph $G$ and only changes its weights. This transformation corresponds to the elimination of the first resistor.}
    \label{fig:relabeling}
\end{figure}

We want to consider the sequence
$$G_f(r_1,r_2) \rightarrow G_f(r_2,r_3) \rightarrow G_f(r_3,r_4) \rightarrow \cdots $$
This sequence of circuits induces the sequence of EqRs $\left| G_f(r_i,r_{i+1}\right|$ given by
$$ r_1 {\cal F}\left(\frac{r_2}{r_1}\right) \rightarrow r_2 {\cal F}\left(\frac{r_3}{r_2}\right) \rightarrow r_3 {\cal F}\left(\frac{r_4}{r_3}\right) \rightarrow \cdots $$
Since $\frac{r_{n+1}}{r_n} \rightarrow \varphi^{-1}$ and $r_n \rightarrow 0$, the EqR approaches $0$ after a large number of relabelings, so the relabel transformation is a contraction in the set of Type 1 FRN. Futhermore, we see that the overall topology does not change and the EqR changes by a scale factor
$s_k = \frac{\left|G_f(r_{k+1},r_{k+2}) \right|}{\left|G_f(r_k,r_{k+1}) \right|} \rightarrow  \frac{1}{\varphi}.$

On another hand, consider the circuit represented by $G_f(r_1,r_2)$ and note that eliminating the first resistor and replacing every parallel set $r_1,r_2$ to the EqR $r_3$, we still have a FRN of type 1, where the first resistor is $r_2$ and the second resistor is now $r_3$. So the relabelling transformation is equivalent to downscaling our network. This illustrates how self-similarity  is  a  key  feature  of  the  FRN.

\section{Some Generalizations}\label{sec:generalizations}

In this section, we analyze some generalizations of the FRN. Although we focus only on generalizations of Type 1 FRN, the analog generalization for Type 2 FRN is straightforward. 

\subsection{FRN Networks of higher order}

Given $p$ initial resistances $r_1,\cdots, r_p$, we can construct a $p$-order FRN using the recurrence series-parallel association depicted in figure \ref{fig:type1-mv}.
\begin{figure}[h!]
    \centering
    \includegraphics[scale=0.7]{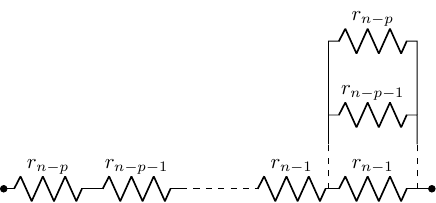}
    \caption{The $n$th term of a type 1 p-order FRN.}
    \label{fig:type1-mv}
\end{figure}

To use the same argument used for the Type 1 FRN, we denote ${\bf r} = (r_1, \cdots, r_p)$ and write $r_k = \frac{r_1}{h_k(\bf r)}$ to see that  
$$ \frac{1}{r_n} = \sum_{k=1}^{k-p} \frac{1}{r_k} \Rightarrow  h_n ({\bf r}) =  \sum_{k=1}^{n-p} h_k({\bf r}).$$ So we find that the numbers $h_k(\bf{r})$ form a $p$-order GFS
$$ h_k({\bf r})  = \frac{r_1}{r_k}, \,\,\,\, 1\leqslant k \leqslant n,\,\,\,\, h_k({\bf r}) = \sum_{j=1}^{p} h_{k-j}({\bf r}) , \forall\,\,k>n+1.$$

As before, we need an expression for the $n$-th term of our $p$-order GFS. The anzats $g_k({\bf r}) = A x^k$ leads to the characteristic equation
\begin{equation}\label{eq:characteristic}
P_p(x)=x^p - x^{p-1} - \cdots - 1   
\end{equation}
The following lemma provides some properties of the roots of the characteristic equation given by \ref{eq:characteristic}. We will show that these roots are located as illustrated in figure \ref{fig:zeros}. It is worth noting that this result is not new. The analysis of the zeros of the characteristic equation given by \ref{eq:characteristic} is of mathematical interest and has already been solved through different methods. For example, \cite{bib:Miles} shows the statement $d$ of the Lemma \ref{l:characteristic} using the Rouchè theorem \cite{bib:Rudin}). In another scenario, \cite{bib:Martin} study the Galois group of \ref{eq:characteristic}. The lemma bellow is included here for completeness and in order to provide a simple and complete analysis of these zeros. 

\begin{figure}[h!]
    \centering
    \includegraphics[scale=0.8]{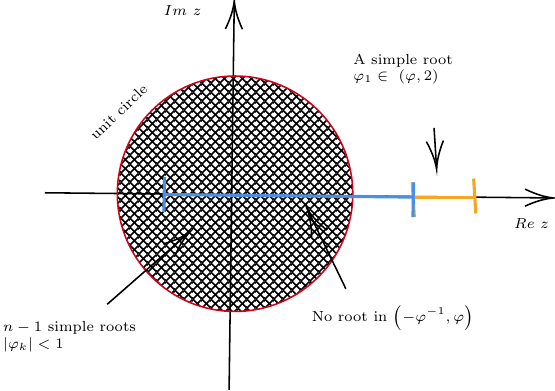}
    \caption{Location of characteristic equation zeros according to lemma  \ref{l:characteristic}}
    \label{fig:zeros}
\end{figure}

\begin{lemma}\label{l:characteristic} The following statements are true for $n>1$
\begin{enumerate}
    \item[(a)] $P_n$ has exactly one positive real root $\varphi_1$, also $\varphi_1$ lies in the open interval $[\varphi,2)$, with $\varphi_1=\varphi$ if and only if $n=2$.
    \item[(b)] If $n$ is odd $P_n$ has no negative real root, and if $n$ is even $P_n$ has exactly one negative real root in $(-1,-\varphi^{-1})$.  
    \item[(c)] All roots of $P_n$ are simple.
    \item[(d)] All other $n-1$ zeros of $P_n$ lies within the unit circle of the complex plane.
    \end{enumerate}
\end{lemma}
\begin{proof}
\begin{enumerate}
    \item[(a)] From Descartes rule of sign $P_n$ has exactly one positive root. Now, note that    $P_2(\varphi)=\varphi^2-\varphi-1=0$, and for $n>2$ we see that
    $$ P_n(\varphi)= \varphi^{n-2}P_2(\varphi)-\sum_{k=0}^{n-2}\varphi^k = -\sum_{k=0}^{n-2} \varphi^k < 0.$$ Furthermore, for $x \neq 1$, $
    P_n(x) = x^n - \frac{x^{n} -1 }{x - 1} $
therefore $P_n(2) = 2^n - 2^n + 1 = 1$. So $\varphi < \varphi_1 <2$. 
\item[(b)] Except for $x=1$ the polynomial
$G_n(x) = (x-1)P_n(x) = x^{n+1} - 2 x^n +1$ has the same set of roots of $P_n$. Now, note that $G_n(-x) = (-1)^{n+1} x^{n+1} - (-1)^n 2 x^n + 1$.  If $n$ is odd, we have $G_n(-x)=x^{n+1} + 2x^n +1$, and therefore and Descartes rule of sign implies that $G_n(-x)$ no positive root, so $G_n(x)$ has no negative root. If $n$ is even, we have $G_n(-x)=- x^{n+1} - 2x^n +1$, and Descartes rule of sign implies that $G_n(-x)$ has one positive root. Note that
for $n$ even, $G_n(-1)=(-1)^{n+1} - 2(-1)^n +1=-2<0$, also denoting $\psi=-\varphi^{-1}$ and that $P_2(\psi)=0$, so we see
$$G_n(\psi)=(\psi-1)\left(\psi^{n-2}P_2(\psi)- \frac{\psi^{n-2}-1}{\psi-1}\right)=1-\psi^{n-2} \geqslant 0$$
where the equality holds for $n=2$.
\item[(c)] Let $q_n(x)$be a $n$ degree polynomial 
with roots $\alpha_j$ ($j=1,\cdots, n$). We can write $q_n (x) = \prod_j ( x- \alpha_j )$, note that its derivative is 
$q'_n(x) = a_n \sum_{k} \prod_{j \neq k} ( x- \alpha_j )$, then $q'(\alpha_i) = a_n \prod_{j\neq i} ( \alpha_i - \alpha_j )$. So if $\alpha_i$ is degenerated then $q'(\alpha_i) =0$. Now, assume $\varphi_j$ is a degenerated root of $G_n$, clearly $\varphi_j \not \in \mathbb{R}$ since all real roots must be simple by (a), (b)  above and $x=0$ is not a root. But $G'(\varphi_j) =0$ implies that \( (n+1) \varphi_j= 2 n\) a contradiction.

\item[(d)] We prove this statement in three parts. First, (i) we show that there is no other zero $\psi$ such that $|\psi|> \varphi_1$. Second, (ii) we prove that there is no zero such that \(1 < |\psi| < \varphi_1\). Finally, (iii) we prove that there is only one zero $\psi$ such that $|\psi|=1$ or $|\psi|=\varphi_1$ that is $\psi = \varphi_1$. The proof of (d) follows from these three statements.

(i) Note that since $P_n(x)$ has exactly one positive real zero $\varphi_1$ and that $\displaystyle \lim_{x \rightarrow +\infty} P_n(x) = +\infty$, then for $x > \varphi_1$,  $P_n(x)>0$. We proceed by contradiction and assume the existence of $\psi \neq \varphi_1$ that is a root of $P_n$ and $|\psi| > \varphi_1$. Since (by hypothesis) $\psi$ is a root then \(\psi^n = \psi^{n-1} + \cdots + \psi + 1\), so the triangle inequality implies that \( |\psi|^n \leqslant |\psi|^{n-1} + \cdots + |\psi| + |1| \), and therefore 
\(P_n(|\psi|) = |\psi|^n - \left( |\psi|^{n-1} + \cdots + |\psi| + |1| \right) \leqslant 0,\)
a contradiction with \(P_n(x)>0\) for $x > \varphi_1$.

(ii) Since $P_n$ has only one positive root $\varphi_1$ and $P_n(x)<0$ for $x\in[0,\varphi_1)$ then $G_n(x) = (x-1)P_n(x)$ has exactly two positive roots $1$ and $\varphi_1$ and $G_n(x) < 0$ for $x \in [0,\varphi_1)$. Now, assume there is a zero $\psi$ of $P_n$ such that \(1 < |\psi| < \varphi_1\). So \(G_n(\psi) = \psi^{n+1} - 2 \psi^n +1 = 0,\) then triangle inequality implies that \( 2|\psi^n| \leqslant |\psi^{n+1}| + 1\). Therefore \(G_n(|\psi|) = |\psi^{n+1}| - 2|\psi^n| + 1 \geqslant 0, \)
a contradiction with \(G_n(x) < 0 \) for for $x \in [0,\varphi_1)$.

(iii) If we have \(G_n\left(e^{i \theta}\right) =0\) then \(e^{i \theta} + e^{- i n \theta} = 2\). Now assume there is $\theta \neq 0$ such that this equality holds, then we must have $\theta \equiv -n \theta + 2 k \pi$, where $k$ is an integer. So we find $\theta_k = \frac{2k}{n+1} \pi$, so there are $n+2$ distinct roots of $G_n$, namely $\varphi_1$ and $n+1$ roots $e^{-i \theta_k}$, a contradiction since $G_n$ only has $n+1$ distinct roots by the fundamental theorem of algebra. So the only root of $G_n(x)$ in the unit circle is $x=1$ and therefore $P_n$ has no roots in the unit circle. If we have \(G_n\left( \varphi_1 e^{i \theta}\right) =0\) then \( \varphi_1 e^{i \theta} + \varphi_1^{-n}e^{- i n \theta} = 2\). Again, if there is $\theta \neq 0$ such that this equality holds, then there are $n+2$ distinct roots of $G_n$, namely $1$ and $n+1$ roots $\varphi_1 e^{-i \theta_k}$, a contradiction.
\end{enumerate}
\end{proof}

Due to Lemma \ref{l:characteristic}.(d), we have the Binet like formula \begin{equation}\label{eq:general-gibonacci}
h_k({\bf r}) = A_1 \varphi_1^{k-1} + A_2 \varphi_2^{k-1} + \cdots + A_p \varphi_p^{k-1},    
\end{equation}
where the constants \(\varphi_k\) ($k=1,\cdots, p$) are the roots of $P_p$ and the components $A_k$ of ${\bf A} = (A_1, \cdots, A_n)$ are functions of \({\bf r}\) given by \({\bf A} = V^{-1} {\bf I} \), where ${\bf I} = \left(1,  \frac{r_1}{r_2}\cdots, \frac{r_1}{r_p}\right)$, and $V$ is the Vandermonde matrix with the $k$ line given by $[\varphi_1^{k-1}\,\,\varphi_2^{k-1}\,\, \cdots \,\, \varphi_p^{k-1}]$. The existence of the inverse $V^{-1}$ is a consequence of zero not be a root of $P_p$ and standard properties of Vandermonde matrices.  

Since the $p$ initial resistances are positive, the recurrence relation implies that $h_k({\bf r})$ is a strictly increasing sequence, so we must have $A_i>0$, $i=1,\cdots,n$. It is clear from Lemma \ref{l:characteristic}.(a,d) that \( r_n = \frac{r_1}{h_k({\bf r})} \rightarrow \frac{r_1}{A_1({\bf r}) \varphi_1^{k-1}} \rightarrow 0\). The following lemma shows that the EqR of the entire network
$$ R_{eq} ({\bf r}) = r_1 \sum_{k} \frac{1}{h_k({\bf r})} = r_1 {\cal G}({\bf r})$$ converges pointwisely. 
\begin{lemma}\label{l:multi-convergence}
The series \({\cal G} \left({\bf r} \right) \) converges pointwisely for all ${\bf r}$ such that $r_k \geqslant 0$, $k=1,...,p$.
\end{lemma}
\begin{proof} 
If any of the $p$ initial resistances are zero, then there is a short circuit in every $n$-th generation with $n \geqslant p+1.$ Therefore
$ {\cal G}\left({\bf r}_0 \right) = \sum_{k=1}^p r_k$ for ${\bf r}_0 \in \{ {\bf r} : r_k = 0, \text{ for some } k=1,\cdots,p\}.$ Otherwise, Lemma \ref{l:characteristic}.d implies that $|\varphi_{i\neq 1}| < 1$, then
$ \frac{r_{n+1}}{r_{n}} = \frac{h_{n}({\bf r})}{h_{n+1}({\bf r})} \rightarrow \frac{1}{\varphi_1} < 1,$
this last inequality is a consequence of Lemma \ref{l:characteristic}.(a). So the ratio test implies the punctual convergence.
\end{proof} 

Finally, we can show the uniform convergence of the p-order FRN.
\begin{theorem} The series \({\cal G} \left({\bf r}\right) \) converges uniformly for ${\bf r} \geqslant 0$
\end{theorem}
\begin{proof}
Let $\hat e_k = (\delta_{1k}, \cdots, \delta_{ip})$ and $w_k (\hat e_i) = \delta_{ik}$, $k=1,\cdots,p$ and $w_{k} (\hat e_i) = \sum_{j=1}^p w_{k-j} (\hat e_i)$ for $k>p$. Each sequence $w_k(\hat e_i)$ is a p-order GFS. Since a linear combination of p-order GFS is a p-order GFS, we see that
\begin{equation}\label{eq:shift}
h_k({\bf r}) = w_k(\hat e_1) + \frac{r_1}{r_2} w_k(\hat e_2) + \frac{r_1}{r_3} w_k(\hat e_3) + \cdots + \frac{r_1}{r_p} w_k(\hat e_p).
\end{equation}
So we can write
$${\cal G} \left({\bf r}\right) = 1 + \frac{r_2}{r_1} + \cdots + \frac{r_p}{r_1} + \sum_{k=p+1}^{+\infty} \frac{1}{w_k(\hat e_1) + \frac{r_1}{r_2} w_k(\hat e_2) + \cdots + \frac{r_1}{r_p} w_k(\hat e_p)}.$$

So for $k>p$, $h(\hat e_i)>0$ and we see that  $\frac{1}{h_{k}({\bf r})} \leqslant \frac{1}{w_k (\hat e_1)}$. On the other hand, $g_k$ is a p-order GFS and therefore$\frac{w_{k}}{w_{k+1}} \rightarrow \frac{1}{\varphi_1} < 1$ as $k\rightarrow +\infty$.
 So $\sum_{k=p}^{+\infty} w_k^{-1}$ converges by the ratio test and therefore  \({\cal G} \left({\bf r}\right) \) converges uniformly for ${\bf r} \geqslant 0$ by Weierstrass M test.

\end{proof}

\subsection{Other Recursive Resistor Networks}

The realm of combinatorics has a whole plethora of important recursive sequences such as the Stirling numbers of the first and second kinds, Bell numbers, central Delannoy numbers, Euler and Genocchi  numbers \cite{bib:Anjos}. The FRN logic of construction can be applied to these sequences to obtain a variety of interesting resistor networks.

As an example, consider the Pell numbers $p_n$ that arise historically in the rational approximation to $\sqrt{2}$, more precisely because
\( \frac{p_{n+1}-p_n}{p_n} \rightarrow \sqrt{2}.\)
For $n\geqslant 2$, we have $p_n = 2p_{n-1}+p_{n-2}$. Following the type 1 logic, we replace each sum by a parallel association and set up each term in a serial association. The recursive procedure to construct a Type 1 Pell Resistor Network  is shown in the inset of figure ~\ref{fig:pell-plot}. Setting $r_1 = 1$, $r_2 = \alpha$, the arguments of the previous sections can be used almost word by word and replacing the golden ratio $\varphi$ to the silver ratio $\phi =\frac{1+\sqrt{2}}{2}$. The resulting behavior is shown in figure. \ref{fig:pell-plot}.
\begin{figure}[h!]
    \centering
    \includegraphics[scale=0.4]{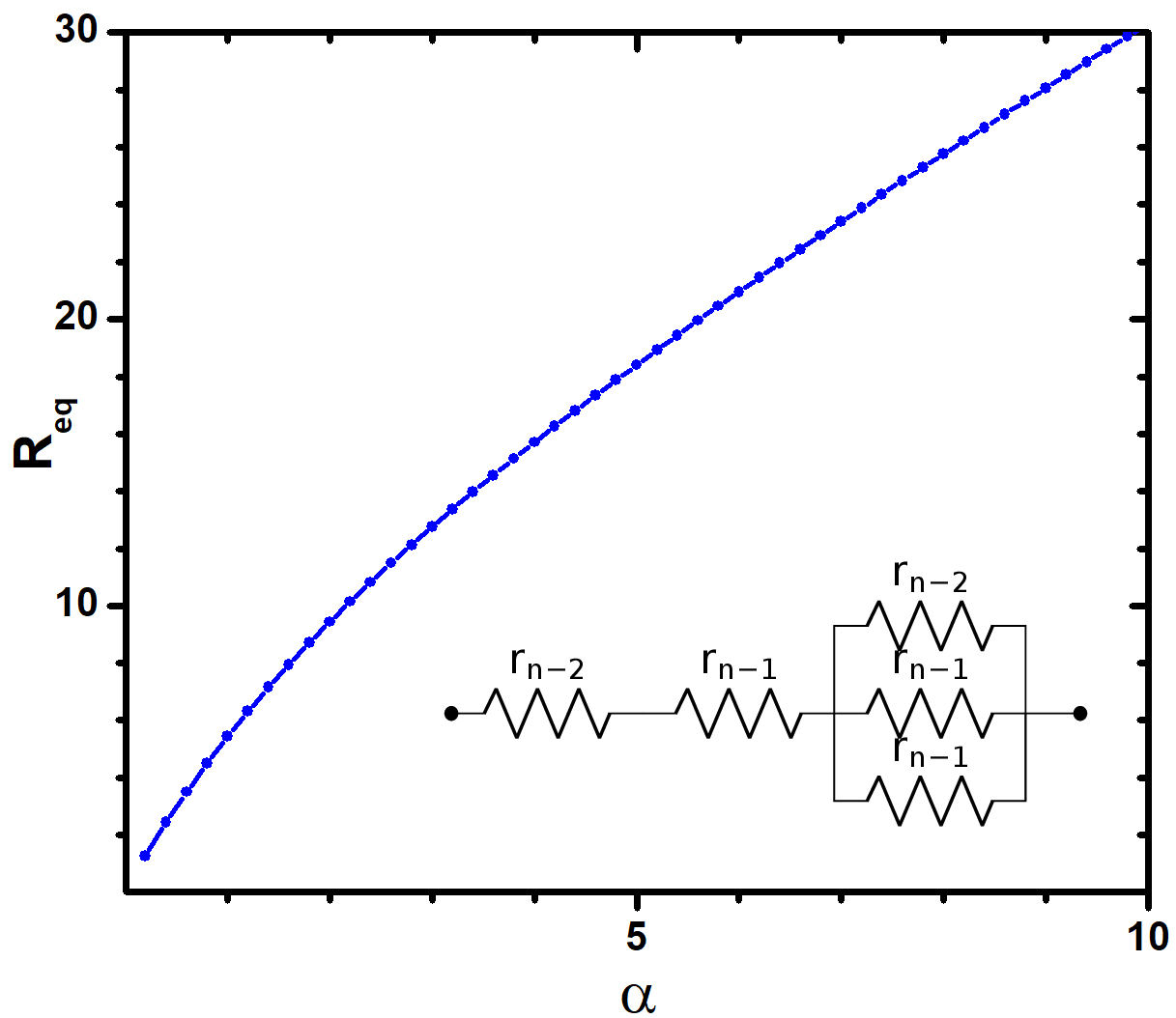}
    \caption{The equilalent resistance of a Type 1 Pell resistor network  as a function of \(\alpha\) for $r_1 = 1$, $r_2=\alpha$. The inset shows the recursive procedure leading to a Type 1 Pell Resistor Network. } 
    \label{fig:pell-plot}
\end{figure}

The Bell sequence gives us another example. The Bell number \(B_n\) counts the number of different ways to partition a set that has exactly \(n\) elements. The Bell sequence follows the recurrence relation \cite{bib:Anjos}
$$B_0=B_1=1,\,\,\,B_{n+1} = \sum_{k=0}^n \frac{n!}{(n-k)! k!} B_k.$$
A possible approach to create a Type 1 Bell Resistor Network is to take $r_{n+1}$ to be a parallel setting of $2^n$ resistors, where for each $k=0,\cdots,n$, we have \(\frac{n!}{(n-k)!k!}\)  resistors $r_k$. The setup for the $n+1$-th group is shown in the inset of figure \ref{fig:bell-plot}. The methods of the previous section does not apply due to the absence of a Binet like formula. Setting $r_0=1$ and $r_1=\alpha$, figure \ref{fig:bell-plot} shows the behavior of the EqR of a Type 1 Bell Resistor Network. 

\begin{figure}[h!]
\centering
\includegraphics[scale=0.41]{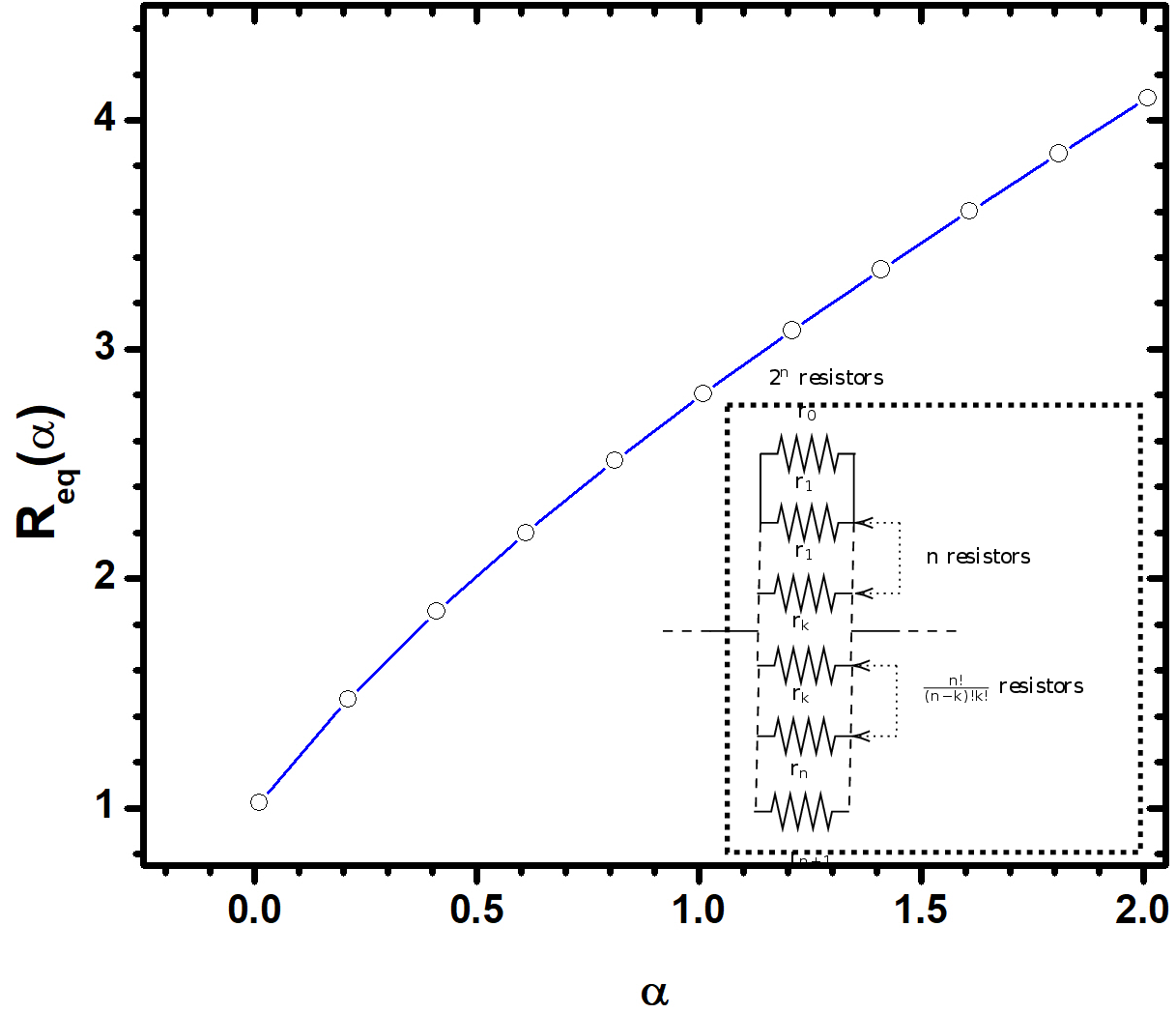}
\caption{Numerical simulation of the EqR of a Type 1 Bell resistor network  as a function of \(\alpha\) for $r_1 = 1$, $r_2=\alpha$ up to $n=20$. The inset shows the construction of the $r_n$ resistance in a Type 1 Bell Resistor Network. } 
\label{fig:bell-plot}
\end{figure}

\section{Conclusions}\label{sec:final}

We propose two kinds of infinite resistor networks based on the  Fibonacci sequence:   parallel of serial(type 1) or serial of parallel(type 2) resistor network. We show that the the network’s EqRs converge uniformly in the parameter $\alpha=\frac{r_2}{r_1}$. This FRN (and the generalizations we discussed)  provides an interesting and non-trivial problem that could be used in Physics and Mathematics teaching.

Furthermore, we show that these networks exhibit self-similarity and scale invariance, which mimics a self-similar fractal. This property suggests that applications in areas that can be investigated with fractal-like networks such as electrical properties of percolation clusters and the electric response of inhomogeneous materials. In fact, these fractal like properties can be connected with a large number of phenomena, for example, recently it was demonstrated that the 2+1 dimensional Kardar-Parisi-Zhang equation exhibits self-similarity and scale invariance with fractal dimension  $d_f=\varphi$, directly related with the growth exponents \cite{GomesFilho21b, Anjos21}.

Moreover,  recent fractal  analysis~\cite{Lima24} of the correlation function in a second  order phase transition close to the critical point shows that the Fisher exponent
\begin{equation}
\label{dff}
 \eta=d-d_f.
\end{equation}
 represents the deviation from the integer dimension. Here $d_f$ is the fractal interface dimension originated by scale invariance~\cite{Lima24}.

\bmhead{Acknowledgments}

 This work was partially financed by the Brazilian Research Agency CNPq. The authors also thank the Distrito Federal Research Foundation FAPDF for financial and equipment support (Edital 04/2017 and Edital 09/2022). D. L. Azevedo acknowledges the support by the Mato Grosso Research Foundation FAPEMAT for financial support through the Grant PRONEX CNPq/FAPEMAT 850109/2009, CENAPAD-SP, and CNPq (Proc. 315623/2021-7-PQ-2). F.A.O. acknowledges the Grant No. 303119/2022-5 and the Funda\c{c}\~ao de Apoio a Pesquisa do Rio de Janeiro (FAPERJ), Grant No. E-26/203953/2022. 

 \bmhead{Data availability statement}
 
All data that support the findings of this study are included within the article.

\section{Declarations}

\textbf{Funding}: David L. Azevedo thanks the Brazilian Research Agencies for financial support: FAPEMAT (PRONEX CNPq/ FAPEMAT 850109/2009), FAP-DF (Edital 04/2017 and Edital 09/2022), and CNPq (Grant No. 315623/2021-7). Fernando A. Oliveira thanks CNPq, grant 303119/2022-5 and FAPERJ, grant E-26/203953/2022 for financial support.
\\
\textbf{Conflict of Interest }: 
The authors declare no competing interests.
\\
\textbf{Author contributions:} These authors contributed equally to this work.


 \bibliography{references}

\end{document}